\documentclass{article}
\usepackage{graphicx} 
\usepackage{fullpage}
\usepackage{times}
\usepackage{amsmath,amsfonts,amsthm,amssymb}
\usepackage{bbm}
\usepackage{algorithm}
\usepackage{algorithmic}
\usepackage{hyperref}
\usepackage{natbib}
\usepackage{nicematrix}

\newcommand{\gap}{\mathrm{Gap}}

\newcommand{\rps}{\mathrm{rps}}

\DeclareMathOperator*{\argmax}{arg\,max}
\DeclareMathOperator*{\argmin}{arg\,min}

\newcommand{\one}{\mathbf{1}}

\newtheorem{theorem}{Theorem} 
\newtheorem{lemma}{Lemma}

\newtheorem{fact}{Fact}

\usepackage[colorinlistoftodos,bordercolor=orange,backgroundcolor=orange!20,linecolor=orange,textsize=scriptsize]{todonotes}

\newcommand{\bvec}[1]{{\overrightarrow{#1}}}

\title{Tie-breaking Agnostic Lower Bound for Fictitious Play}
\author{Yuanhao Wang\\
Princeton University\\
\texttt{yuanhao@princeton.edu}}
\date{}

\begin{document}

\maketitle

\begin{abstract}
Fictitious play (FP) is a natural learning dynamic in two-player zero-sum games. Samuel Karlin conjectured in 1959 that FP converges at a rate of $O(t^{-1/2})$ to Nash equilibrium, where $t$ is the number of steps played. However, Daskalakis and Pan disproved the stronger form of this conjecture in 2014, where \emph{adversarial} tie-breaking is allowed.

This paper disproves Karlin's conjecture in its weaker form. In particular, there exists a $10$-by-$10$ zero-sum matrix game, in which FP converges at a rate of $\Omega(t^{-1/3})$, and no ties occur except for the first step.
\end{abstract}

\section{Introduction}
A two–player zero–sum game can be written as the bilinear saddle–point problem
\begin{equation}
\label{eq:minimax}
    \max_{x\in\Delta_n}\min_{y\in\Delta_m} x^{\top}Ay,
\end{equation}
where $A\in[-1,1]^{n\times m}$ is the payoff matrix and $\Delta_k$ denotes the $k$–dimensional simplex.  Von~Neumann's celebrated minimax theorem \citep{von1928theorie} guarantees the value of~\eqref{eq:minimax} and the existence of a \emph{Nash equilibrium} $(x^{\star},y^{\star})$ satisfying
\[
    x^{\top}Ay^{\star}\;\le\;(x^{\star})^{\top}Ay^{\star}\;\le\;(x^{\star})^{\top}Ay, \qquad \forall\,x\in\Delta_n,\;y\in\Delta_m.
\]
Computing such an equilibrium remains a central problem in game theory, online learning, and optimization.

One of the earliest algorithms for computing a Nash equilibrium is \emph{fictitious play} (FP)~\citep{brown1951iterative}. FP is an iterative procedure, where in every iteration, each player plays the best response to the {empirical} mixture of the opponent's past actions. More specifically, the update rule at round $t$ is~\footnote{Here the initial points are $x_0=\mathbf{0}$, $y_0=\mathbf{0}$. Note that $\Vert x_t\Vert_1=\Vert y_t\Vert_1=t$.}
\begin{align}
\label{eq:fp1}
\begin{cases}
    i_t &\gets \argmax_i (Ay_{t-1})[i]\\
    x_{t} &\gets x_{t-1} + e[i_t]\\
    j_t &\gets \argmin_j (A^\top x_{t-1})[j]\\
    y_{t} &\gets y_{t-1} + e[j_t]
\end{cases}.
\end{align}
The averaged play $(x_t/t,\,y_t/t)$ is guaranteed to converge to Nash equilibrium, but the {rate} of convergence is delicate and, in spite of seventy years of study, still not completely understood.

The classical convergence analysis~\citep{robinson1951iterative,shapiro1958note} implies a convergence rate of $O\bigl(t^{-1/(m+n-2)}\bigr)$. Meanwhile,~\citet{karlin1959mathematical} conjectured the much faster $O\bigl(t^{-1/2}\bigr)$ rate.  A seminal counterexample by \citet{daskalakis2014counter} refuted Karlin's conjecture under \emph{adversarial tie–breaking}; their construction achieves $\Omega\bigl(t^{-1/n}\bigr)$ convergence in the $n\times n$ identity matrix. However, their construction heavily relies on a time-varying adaptive tie-breaking rule. In fact, it has been shown by~\citet{abernethy2021fast} that for diagonal payoff matrices, which include the identity matrix, the convergence rate is upperbounded by $O(t^{-1/2})$ under \emph{lexicographic tie-breaking}. 

Consequently, the following weaker form of Karlin's conjecture has remained open:
\begin{quote}
    \emph{Does fictitious play converge at the $O(t^{-1/2})$ rate when ties are broken lexicographically, or when no ties occur at all?}
\end{quote}

This paper resolves the question negatively. A $10\times 10$ matrix game is constructed in which FP, regardless of the tie-breaking rule, has a duality gap lower bounded by $\Omega(t^{-1/3})$ at the $t$-th iteration.


\subsection{Related work}

\paragraph{Analysis of Fictitious Play} \citet{harris1998rate} gave an alternative proof of FP convergence by analyzing the continuous-time version of fictitious play. However, the proof is asymptotic in nature and does not imply a finite time rate. The fast $O(1/t)$ rate of continuous-time fictitious play also highlights a large gap between the continuous and discrete time algorithms. For general-sum games, FP is not guaranteed to converge to the Nash equilibrium ~\citep{gaunersdorfer1995fictitious}, except for special classes of games, such as $2\times n$ games~\citep{berger2005fictitious}.

\paragraph{Best-response oracle} FP belongs to a class of game-solving algorithms that access the game through best-response oracles. Indeed, as observed by~\citet{gidel2017frank}, fictitious play can be interpreted as an online Frank-Wolfe method with step size $1/(t+1)$. This property made the algorithm amenable to extensions to learning Nash equilibrium in extensive-form games~\citep{heinrich2015fictitious} and with neural networks~\citep{heinrich2016deep}. 


\section{Preliminaries}
\subsection{Notations}
For a pair of mixed strategies $(x,y)$,  the duality gap is defined as
\[
\gap_A(x,y):= \max_{\hat x\in\Delta_n}\hat{x}^\top Ay - \min_{\hat y\in\Delta_m} x^\top A \hat{y}.
\]
The subscript will be omitted when the payoff matrix referred to is clear from context. It is known that $\gap(x,y)\ge 0$ and $\gap(x,y)=0$ if and only if $(x,y)$ is a Nash equilibrium.

For a vector $w$, we use $w[i]$ to refer to its $i$-th entry. We will use $\max\{w\}$ to denote $\max_i w[i]$, and $\argmax\{w\}$ to denote $\argmax_i w[i]$. $t$ and $k$ will be reserved for integer indices. In addition, we define $\bvec{k}:=[k^2;k;1]^\top$ to simplify the presentation of quadratic polynomials.

\subsection{Symmetric games}
Symmetric games are zero-sum matrix games where $A=-A^\top$. In a symmetric game, assuming the same tie-breaking method is applied to both $x$ and $y$ players, one could immediately see that $x_t=y_t$ for all $t$. Therefore, FP can be simplified as (for $t=1,\cdots$):
\begin{align}
\label{eq:fp-sym}
\begin{cases}
    i_t &\gets \arg\max \{Ax_{t-1}\}\\
    x_{t} &\gets x_{t-1} + e[i_t]
\end{cases}.
\end{align}
Equivalently, one can track the change of $U_t:=Ax_{t}$:
\begin{align}
\label{eq:fp-sym-U}
\begin{cases}
    i_t &\gets \arg\max \{U_{t-1}\}\\
    U_{t} &\gets U_{t-1} + A[:,i_t]
\end{cases}.
\end{align}
The construction presented in this paper is a symmetric game. Thus, the simplified updates (\ref{eq:fp-sym}) and (\ref{eq:fp-sym-U}) will always be used. In a symmetric game, the duality gap can be expressed as
\[
\gap_A(x,x) = \max\{Ax\} - \min\{-Ax\} = 2\max\{Ax\}.
\]
This enables us to use (\ref{eq:fp-sym-U}) and track the convergence rate through the growth of $\max\{U_t\}$: if $\max\{U_t\}=\Theta(t^a)$, then $\gap(x_t/t, x_t/t) = \Theta(t^{a-1})$.

An important fact about FP that is made obvious by (\ref{eq:fp-sym-U}) is that $\max\{U_t\}$ grows monotonically.
\begin{fact}
\label{fact:monotone}
$\forall t\ge 0$, $\max\{U_{t+1}\} \ge \max\{U_{t}\}$. In particular, $\max\{U_{t+1}\} > \max\{U_t\}$ only if $i_{t+2}\neq i_{t+1}$.
\end{fact}
\begin{proof}
Note that
\begin{align*}
\max\{U_{t+1}\}=U_{t+1}[i_{t+2}] \ge U_{t+1}[i_{t+1}] = U_t[i_{t+1}] + A[i_{t+1},i_{t+1}] = \max\{U_t\}.
\end{align*}
Here $A[i_{t+1},i_{t+1}]=0$ since $A=-A^\top$. 
\end{proof}

\section{The Rock-Paper-Scissors game}
The classic Rock-Paper-Scissors (RPS) game is characterized by the $3\times 3$ payoff matrix
\[
A_{\rps} = \begin{pmatrix}
    0 & -1 & 1\\
    1 & 0 & -1\\
    -1 & 1 & 0
\end{pmatrix}.
\]
If FP (\ref{eq:fp-sym-U}) is run in RPS with \emph{lexicographic tie-breaking} ($1\succ 2\succ 3$), the actions $i_t$ will follow the following pattern:
\[
1,2,2,2,3,3,3,3,3,\cdots,\underbrace{1,\cdots,1}_{6k+1},\underbrace{2,\cdots,2}_{6k+3},\underbrace{3,\cdots,3}_{6k+5},\cdots
\]
This pattern can be formalized as follows.
\begin{fact}
\label{fact:rps}
Let $x_t^{\rps}$ and $U_t^{\rps}$ be FP iterates for RPS under lexicographic tie-breaking. $\forall k\ge 1$,
\begin{align}
x_{9k^2}^{\rps} &= [3k^2-2k, 3k^2, 3k^2 + 2k]^\top, \qquad &U_{9k^2}^{\rps} &= [2k, -4k, 2k]^\top,\label{eq:vertex1}\\
x_{9k^2+6k+1}^{\rps} &= [3k^2+4k+1, 3k^2, 3k^2 + 2k]^\top, &U_{9k^2+6k+1}^{\rps}& = [2k, 2k+1, -4k-1]^\top,\label{eq:vertex2}\\
x_{9k^2+12k+4}^{\rps} &= [3k^2+4k+1, 3k^2+6k+3, 3k^2 + 2k]^\top, &U_{9k^2+12k+4}^{\rps} &= [-4k-3, 2k+1, 2k+2]^\top.\label{eq:vertex3}
\end{align}
\end{fact}
Intuitively, this periodicity shows that in RPS, the average strategy $x_t$ evolves in outward-spinning triangles, and (\ref{eq:vertex1})-(\ref{eq:vertex3}) captures moments the iterate is at one vertex of the triangle. We will be heavily exploit this periodicity. Recall that $\bvec{k}=[k^2;k;1]^\top$. Then the Fact above can be conveniently rewritten in the following way:~\footnote{For the curious reader, $Q$ stands for ``quadratic''.}
\[
\text{For }Q\in\left\{\begin{pmatrix}
3 & -2 & 0 \\
3 & 0 & 0 \\
3 & 2 & 0
\end{pmatrix}, \begin{pmatrix}
3 & 4 & 1 \\
3 & 0 & 0 \\
3 & 2 & 0
\end{pmatrix},\begin{pmatrix}
3 & 4 & 1 \\
3 & 6 & 3 \\
3 & 2 & 0
\end{pmatrix}\right\}, \forall k\ge 1, \text{when } t=\one^\top Q\bvec{k}, x_t^\rps = Q\bvec{k}.
\]
We define \emph{admissible} matrices for RPS as those that satisfy this property, namely for all $k$, when $t=\one^\top Q\bvec{k}$, $x_t=Q\bvec{k}$. By interpolating between the three admissible matrices given in Fact~\ref{fact:rps}, one can obtain many more admissible matrices. In particular, the following two admissible matrices will be used in the hard instance construction:
\begin{equation}
\label{eq:qmat}
    Q_0 = 
\begin{pmatrix}
12 & 8 & -12 \\
12 & 0 & 0 \\
12 & 4 & 0
\end{pmatrix}
\quad\text{and}\quad
Q_1 =
\begin{pmatrix}
48 & 208 & 225 \\
48 & 200 & 213 \\
48 & 200 & 208
\end{pmatrix}.
\end{equation}
\begin{fact}
$Q_0$ and $Q_1$ are admissible for RPS.
\end{fact}
\noindent This can be seen by applying Fact~\ref{fact:rps} to $2k$ and $4k+8$ respectively and interpolating the vertices. 

A further observation is that if fictitious play is initialized with $U_0 = A_{rps}Q_0\bvec{k}$, then after $t=\one^\top (Q_1-Q_0) \bvec{k}$ steps, $x_t=(Q_1-Q_0)\bvec{k}$. The reason is that the FP with initialization evolves as if FP without initialization from $t_0=\one^\top Q_0\bvec{k}$ to $t_1=\one^\top Q_1\bvec{k}$.

\section{A Hard Instance for Fictitious Play}
\label{sec:hard-instance}
This section gives the full construction of the hard $9\times9$ matrix game that attains the $\tilde\Omega(t^{-1/3})$ lower bound. Building on the periodic structure exhibited in the $3\times 3$ rock-paper-scissors game, we embed three coupled copies into a larger $9\times 9$ game so that FP follows a double-loop trajectory:
\begin{enumerate}
    \item an inner loop in which FP plays one of the RPS instances for $\Theta(t^{2/3})$ steps;
    \item an outer where the ``active'' RPS copy transitions from block~1 to block~2 to block~3, and back to block~1.
\end{enumerate}
The benefit of a double-loop structure is that the three RPS instances will interact with each other in a way that ``cancels'' the progress of each other. This ensures that in every outer loop block, $\max\{U_t\}$ increases by $\sqrt{t^{2/3}} = t^{1/3}$ -- as if the FP was re-initialized every time the RPS instance is played.

This section is organized as follows. Sec~\ref{sec:hard-instance-derivation} will provide the motivation for the sufficient conditions of a counterexample. Sec~\ref{sec:hard-instance-number} presents one solution of the sufficient condition. Sec~\ref{sec:hard-instance-proof} will prove the lower bound.

\subsection{Deriving the instance}
\label{sec:hard-instance-derivation}
The $9\times 9$ game will be a block matrix that contains three RPS instances that interact with each other in a cyclical symmetric way:
\begin{equation}
    \label{eq:Mdef}
    M = \begin{pmatrix}
    A_{\rps} & B & -B\\
    -B & A_{\rps} & B\\
    B & -B & A_{\rps}
\end{pmatrix}.
\end{equation}
To ensure that $M=-M^\top$, the $3\times 3$ interaction matrix $B$ needs to be symmetric. This is the key part of the hard instance which needs to be designed such that FP on $M$ has a double-loop phased structured.

To be more precise, the desired effect is a phased structure shown in Fig.~\ref{fig:phase-structure}. Here $Q$, $V_1$, $V_2$, $V_3$ are constant $3\times 3$ matrices to be determined. Starting from some time $T_k$, in time period $[T_k, T_{k+1})$, all actions played will be in the first block and can be described by $Q\bvec{k}$; in time period $[T_{k+1}, T_{k+2})$, all actions played will fall in the second block; so on and so forth. The evolution of $U_t$ will have a synchronized cyclic structure, so that in time period $[T_k, T_{k+1})$, its maximum always fall in the first block.

\begin{figure}[ht]
    \centering
    \includegraphics[width=0.9\linewidth]{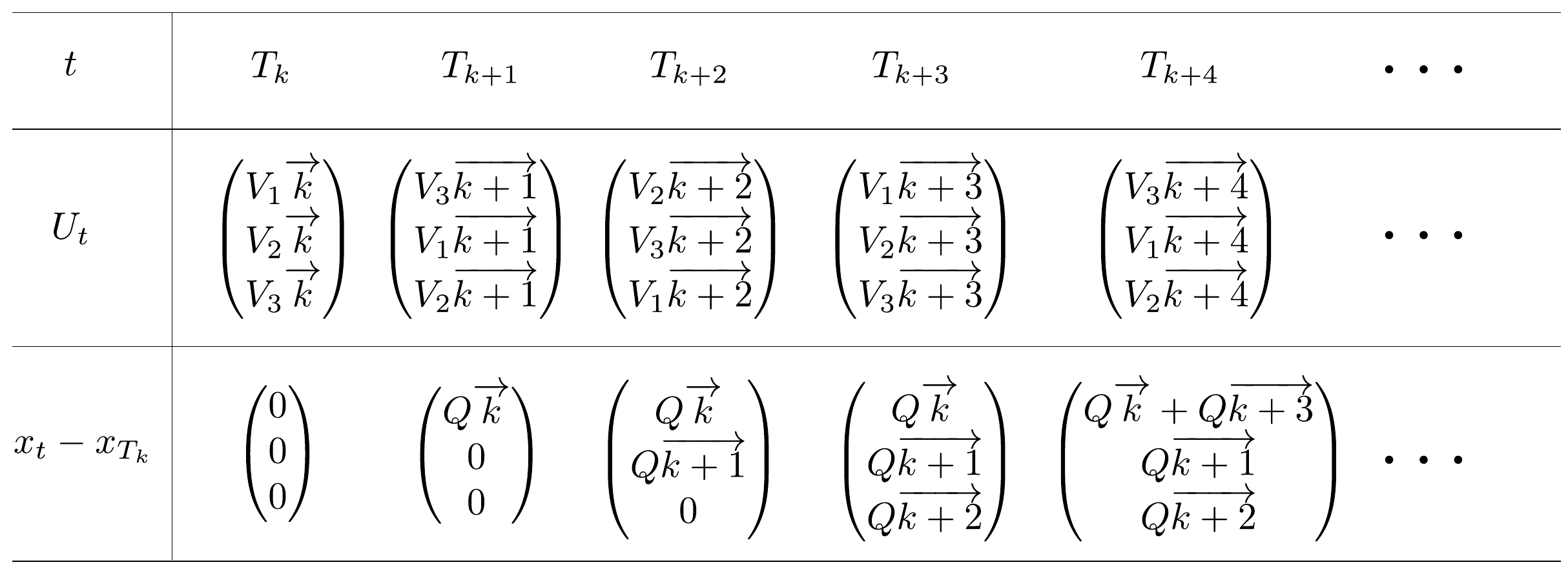}
    \caption{Phased Structure. It is assumed that $k\mod 3=1$ here for illustration.}
    \label{fig:phase-structure}
\end{figure}

It is straightforward to see that if the phased structure in Fig.~\ref{fig:phase-structure} holds, $\max\{U_{T_k}\}$ will be $O(k^2)$ while $T_k$ is $\Theta(k^3)$, which proves the desired $\Omega(t^{-1/3})$ lower bound.

In order for the phased structure to hold, a necessary condition is that for all $k$,
\[
 \begin{pmatrix}
    V_1\bvec{k}\\V_2\bvec{k} \\ V_3\bvec{k}
\end{pmatrix} + \begin{pmatrix}
    A_{\rps} & B & -B\\
    -B & A_{\rps} & B\\
    B & -B & A_{\rps}
\end{pmatrix} \begin{pmatrix}
    Q\bvec{k}\\0 \\0
\end{pmatrix} =  \begin{pmatrix}
    V_3 \bvec{k+1}\\ V_1\bvec{k+1} \\ V_2\bvec{k+1}
\end{pmatrix}.
\]
Using the fact that 
\[
\bvec{k+1} = \begin{pmatrix}
   (k+1)^2 \\ k+1 \\ 1
\end{pmatrix}= \begin{pmatrix}
    1 & 2 & 1\\
    0 & 1 & 1\\
    0 & 0 &1
\end{pmatrix}\bvec{k} =:  C\bvec{k},
\]
the condition above can be rewritten in matrix form:
\begin{align}
\label{eq:blocks}
\begin{cases}
        V_1 + A_{\rps}Q &= V_3 C\\
    V_2 - BQ &= V_1 C\\
    V_3 + BQ &= V_2 C
\end{cases}.
\end{align}

In order to ensure that whenever a block's turn is up, the next $\one^\top Q\bvec{k}$ actions taken will indeed be described by $Q\bvec{k}$, we will choose $Q=Q_1-Q_0$, and choose $V_1$ such that
\[
V_1\bvec{k} = \Delta(k)\one + A_{\rps}Q_0\bvec{k}.
\]
This will ensure that within every phase, the evolution of $U$ within the block -- up to a constant -- is exactly the evolution within RPS from time $\one^\top Q_0\bvec{k}$ $\one^\top Q_1\bvec{k}$. 
Here $\Delta(k)$ is a scalar quadratic function of $k$. The choice of $\Delta$ will be given later and does not affect the evolution within the block since it's a constant for each action.

The equation above can be written in matrix form as 
\begin{equation}
\label{eq:init}
V_1 = \one \Delta^\top + A_{\rps} Q_0. 
\end{equation}
By expressing $V_1$, $V_2$ and $V_3$ with $B$ and $\Delta$, we can combine (\ref{eq:blocks}) and (\ref{eq:init}) into one matrix equation about $A$ and $\Delta$:
\begin{equation}
\label{eq:key}
A_{\rps}Q + BQ(C-C^2) = \one\Delta^\top (C^3-I) + A_{\rps}Q_0(C^3-I).
\end{equation}
This is the key equation that we need to solve.

However, equation (\ref{eq:key}) itself is necessary but not sufficient for the phased structure illustrated in Fig.~\ref{fig:phase-structure}: we also need  the outer loop, \emph{i.e.} the transition between blocks, to be timed exactly. That is, we need to ensure that for any $t\in [T_k, T_{k+1})$,
\[
\max\{U_t[1:3]\} > \max\{ U_t[4:9]\},
\]
so that the first block's evolution is not ``interrupted'' prematurely. The tactic for enforcing this condition is illustrated by Fig.~\ref{fig:block-max}: we will design $B$ to be fully negative, so that when the first block is played, the maximum within the second block increases monotonically while the maximum within the third block decreases monotonically. Then, by requiring
\begin{equation}
\label{eq:key-ineq}
    0<\max\{V_1 \bvec{k}\} - \max\{V_3\bvec{k}\}<-\max\{B\},
\end{equation}
we can ensure that the second block catches up with the first block at just the right time: $\max_t\{U_t[4:6]\}$ will be smaller than $\max_t\{U_t[1:3]\}$ at $t=T_{k+1}-1$, but greater than it at $t=T_{k+1}$.

\begin{figure}[ht]
    \centering
    \includegraphics[width=0.95\linewidth]{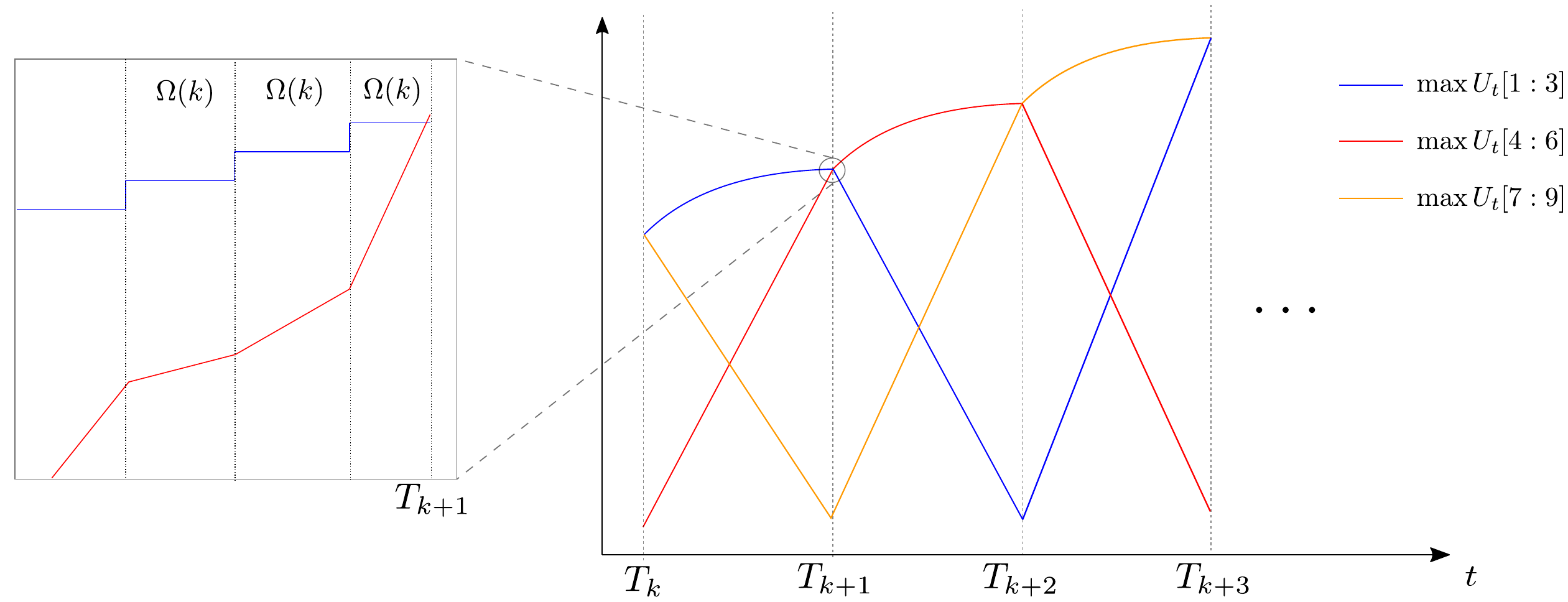}
    \caption{Evolution of the maximum within each block.}
    \label{fig:block-max}
\end{figure}

\subsection{The instance}
\label{sec:hard-instance-number}
At this point, the task for finding a lower bound essentially boils down to: find a pair of admissible matrix $Q_0$, $Q_1$, such that $B$ and $\Delta$ can be solved from (\ref{eq:key}) and satisfy (\ref{eq:key-ineq}). A simple brute-force search is used to identify one such solution, given by admissible matrices defined in (\ref{eq:qmat}), and
\[
B = -\frac{1}{900}\begin{pmatrix}
    71 & 54 & 75\\
    54 & 21 & 25\\
    75 & 25 & 50
\end{pmatrix}, \quad \Delta=\begin{pmatrix}
    2\\
    \frac{290}{27}\\
    0
\end{pmatrix}.
\]
It can be easily checked that $B$ and $\Delta$ satisfy (\ref{eq:key}). The corresponding $V_1$, $V_2$, $V_3$ are found as
\[
V_1 = \frac{1}{27}\begin{pmatrix}
   54 & 406 & 0\\
   54 & 406 & -324\\
   54 & 82 & 324
\end{pmatrix},\quad
V_2 = -\frac{1}{2700}\begin{pmatrix}
    16200 & 67700 & 85787 \\
    5400 & 8300 & 53813 \\
    10800 & 70400 & 54500
\end{pmatrix},\quad
V_3 = \frac{1}{27}\begin{pmatrix}
    54 & 190 & -379\\
    54 & 406 & -1\\
    54 & -26 & -352
\end{pmatrix}.
\]
It can then be checked that $\forall k$,
\[
\max\{V_1 \bvec{k}\} - \max\{V_3\bvec{k}\} = \frac{1}{27} \in (0, -\max\{B\}).
\]
The counterexample could be instantiated by running FP on $M$ (defined in (\ref{eq:Mdef})) with lexicographic tie-breaking and initialization 
\[
U_0 = \begin{pmatrix}
    V_1\bvec{1}\\
    V_2\bvec{1}\\
    V_3\bvec{1}
\end{pmatrix} = \left[\frac{460}{27}, \frac{136}{27}, \frac{460}{27}, -\frac{169687}{2700}, -\frac{67513}{2700}, -\frac{1357}{27}, -5, 17, 12\right]^\top.
\]
These two caveats can be removed by considering a $10\times 10$ game, augmented with a dummy action that will only be played in the first step. In particular, define
\[
\hat{U}_0 = U_0 + \frac{169687}{2700}\one + \left[2\delta,\delta,0,2\delta,\delta,0,2\delta,\delta,0\right]^\top,
\]
where $\delta$ can be chosen as an arbitrary constant in $(0, 1/1800)$. Then a $10\times 10$ augmented game can be defined as:
\begin{equation*}
     {M}_{aug} = \begin{pNiceArray}{c|ccc}
     0 & \Block{1-3}{-\hat{U_0}^\top}\\\hline
\Block{3-1}{\hat{U_0}} & A_{\rps} & B & -B\\
    &-B & A_{\rps} & B\\
    &B & -B & A_{\rps}
\end{pNiceArray}.
\end{equation*}
Since $U_0=\emph{0}$, a tie necessarily occurs, and it is without loss of generality to assume that $i_1=1$. In that case, the first FP step in $M_{aug}$ will result in $U_1 = [0;\hat{U}_0]$. Since
\[
\max\{\hat{U}_0\} > \min\{\hat{U}_0\} > 0, 
\]
the first action will never be chosen by FP again. Therefore, since the second step, the execution of FP on $M_{aug}$ will be identical to FP on $M$ with initial vector $U_0$ and lexicographic tie-breaking.

\begin{lemma}
\label{lem:init}
Denote the sequence of actions that FP takes with lexicographic tie-breaking and initial vector $U_0$ by $\{i_t\}$. Denote the sequence of actions of FP in $M_{aug}$ by $\{\hat{i}_t\}$. Then for all $t\ge 1$, $i_t = \hat{i}_{t+1}$.
\end{lemma}

\subsection{Lower bound proof}
\label{sec:hard-instance-proof}
We now provide a formal proof that the $9\times 9$ matrix game defined in Sec~\ref{sec:hard-instance-number} indeed shows the desired $\Omega(t^{-1/3})$ lower bound. Define 
\[
T_k = \sum_{j=1}^{k-1} \one^\top Q \bvec{j} = 36k^3 + 244k^2 + 378k - 658.
\]
It follows that $T_{k+1}-T_k = \one^\top Q \bvec{k}$. The following two technical lemmas formalize the argument made in Sect.~\ref{sec:hard-instance-derivation} that (\ref{eq:key-ineq}) is sufficient. Lemma~\ref{lem:block} shows that $\max\{U_t[1:3]\}>\max\{U_t[4:9]\}$ assuming the steps in $[T_k,t)$ are made according to the RPS trajectory. Lemma~\ref{lem:oneblock} then applies induction to show that the first block will follow RPS trajectory throughout the phase $[T_k,T_{k+1}]$.
\begin{lemma}
\label{lem:block}
For any $k\ge 1$, define
\[
t_0=\one^\top Q_0\bvec{k} = 36k^2+12k-12, t_1=\one^\top Q_1\bvec{k}=44k^2+608k+646.
\]
Then for all $\tau\in [0, T_{k+1}-T_k)$,
\begin{align}
   \max\{ V_1\bvec{k} + A_{\rps} \left(x^\rps_{t_0  + \tau} - x^{\rps}_{t_0}\right)\}&> \max\{V_2\bvec{k} - B \left(x^\rps_{t_0  + \tau} - x^{\rps}_{t_0}\right)\}, \label{eq:comp12}\\
    \max\{ V_1\bvec{k} + A_{\rps}\left(x^\rps_{t_0  + \tau} - x^{\rps}_{t_0}\right)\}&> \max\{V_3\bvec{k} + B \left(x^\rps_{t_0  + \tau} - x^{\rps}_{t_0}\right)\}. \label{eq:comp23}
\end{align}
\end{lemma}
\begin{proof}
First off, because
\[
V_1\bvec{k} = (\one \Delta^\top + A_{\rps} Q_0)\bvec{k} = \Delta(k)\one  + A_{\rps} x^\rps_{t_0},
\]
we have
\[
V_1\bvec{k} + A_{\rps} \left(x^\rps_{t_0  + \tau} - x^{\rps}_{t_0}\right) = \Delta(k)\one + A_{\rps} x^\rps_{t_0 + \tau}.
\]
It follows immediately from Fact~\ref{fact:monotone} that
\[
\max\{\Delta(k)\one + A_{\rps} x^\rps_{t_0 + \tau}\} \ge  \max\{\Delta(k)\one + A_{\rps} x^\rps_{t_0}\} = \max\{ V_1\bvec{k}\}.
\]
Therefore (\ref{eq:comp23}) follows as
\[
\max\{V_1\bvec{k}\} > \max\{V_3\bvec{k}\} \ge \max\{V_3\bvec{k} + B \left(x^\rps_{t_0  + \tau} - x^{\rps}_{t_0}\right)\}.
\]
Here the second equality holds as $ \left(x^\rps_{t_0  + \tau} - x^{\rps}_{t_0}\right)[i]\ge 0$ for all $i\in \{1,2,3\}$ while $\max\{B\}<0$.

The proof of (\ref{eq:comp12}) is slightly more involved. Notice that at the end of the interval, {\emph{i.e.}} when $\tau = T_{k+1}-T_k$, the left hand side becomes
\[
\max\{V_1\bvec{k} + A_{\rps}Q\bvec{k}\} = \max\{V_3\bvec{k+1}\} = 2k^2 + \frac{514}{27}k + 17,
\]
while the right hand side becomes
\[
\max\{V_2\bvec{k} - BQ\bvec{k}\} = \max\{ V_1 \bvec{k+1}\} = 2k^2 + \frac{514}{27}k + \frac{460}{27} = \max\{V_1\bvec{k} + A_{\rps}Q\bvec{k}\} + \frac{1}{27}.
\]
We can then compare the value of both sides by estimating the difference for $\tau\in [0, T_{k+1}-T_k)$ and $\tau = T_{k+1}-T_k$. In particular, since $-\max\{A\} = \frac{1}{12} > 0$, the right hand side of  (\ref{eq:comp12}) grows by at least $\frac{1}{12}$ when $\tau$ increases by $1$. In other words, at time $\tau$,
\[
{\rm RHS} \le \max\{V_1\bvec{k} + A_{\rps}Q\bvec{k}\} + \frac{1}{27} - \frac{1}{12}(T_{k+1} - T_k - \tau).
\]
Meanwhile, the left hand side grows with $\tau$ in a $\Theta(1/k)$ rate. Between $\tau \in [T_{k+1}-8k-21, T_{k+1})$,
\begin{align*}
 {\rm LHS} = \max\{V_1\bvec{k} + A_{\rps}Q\bvec{k}\} > \max\{V_1\bvec{k} + A_{\rps}Q\bvec{k}\} + \frac{1}{27} - \frac{1}{12}(T_{k+1} - T_k - \tau) \ge {\rm RHS}.
\end{align*}
When $\tau < T_{k+1} - 12(4k+17) - 1$,
\begin{align*}
{\rm LHS} &\ge \max\{V_1\bvec{k}\} = \max\{V_1\bvec{k} + A_{\rps}Q\bvec{k}\} - (4k+17)\\
&\ge  \max\{V_1\bvec{k} + A_{\rps}Q\bvec{k}\} - \frac{1}{12}(T_{k+1} - T_k - \tau)  + \frac{1}{27} = {\rm RHS}.
\end{align*}
Finally, when $\tau \in [T_{k+1} - 12(4k+17) -1, T_{k+1}-8k-21)$,
\begin{align*}
{\rm LHS} &\ge \max\{V_1\bvec{k}+A_{\rps}Q\bvec{k}\} - 1 - \frac{1}{24k}\left(T_{k+1}-\tau-1\right) > \\
&\ge \max\{V_1\bvec{k}+A_{\rps}Q\bvec{k}\} - \frac{1}{12}\left(T_{k+1}-\tau\right) + \frac{1}{27} \ge {\rm RHS}.
\end{align*}
Putting the three cases together proves (\ref{eq:comp12}).
\end{proof}

\begin{lemma}
\label{lem:oneblock}
Assume that for some $k$, when $t=T_k$,
\[
U_t =  \begin{pmatrix}
    V_1\bvec{k}\\
    V_2\bvec{k}\\
    V_3\bvec{k}\\
\end{pmatrix},
\]
then $\tau\in [0, T_{k+1} - T_{k}]$,
\[
x_{T_k+\tau}-x_{T_k} = \begin{pmatrix}
x^\rps_{t_0+\tau} - x^\rps_{t_0}    \\
0\\
0\\
\end{pmatrix}
\]
\end{lemma}
\begin{proof}
We prove this lemma via induction on $\tau$. The base case ($\tau=0$) obviously holds. By Lemma~\ref{lem:block}, if $\tau<T_{k+1}-T_k$ and
\[
x_{T_k+\tau}-x_{T_k} = \begin{pmatrix}
x^\rps_{t_0+\tau} - x^\rps_{t_0}    \\
0\\
0\\
\end{pmatrix},
\]
then the maximum of $U_{T_k+\tau}$ must be in the first block. Therefore
\[
i_{T_k+\tau+1} = i^{\rps}_{t_0+\tau+1},
\]
which implies that
\[
x_{T_k+\tau+1}-x_{T_k} = \begin{pmatrix}
x^\rps_{t_0+\tau+1} - x^\rps_{t_0}    \\
0\\
0\\
\end{pmatrix}.
\]
Induction on $\tau$ from $0$ to $T_{k+1}-T_k-1$ proves the lemma.
\end{proof}

\begin{theorem}
When FP on $M$ is initialized with $U_0$, for any $k\ge 1$, when $t=T_{3k+1}=\Theta(k^3)$,
\[
U_{t} = \begin{pmatrix}
V_1 \bvec{3k+1}\\
V_2 \bvec{3k+1}\\
V_3 \bvec{3k+1}\\
\end{pmatrix},
\]
Thus $\gap(x_t/t, x_t/t) = 2\max\{U_t\}/t = \Theta(t^{-1/3})$.
\end{theorem}
\begin{proof}
By setting $\tau = T_{k+1}-T_k$, Lemma~\ref{lem:oneblock} implies that if 
\[
U_{T_k} =  \begin{pmatrix}
    V_1\bvec{k}\\
    V_2\bvec{k}\\
    V_3\bvec{k}\\
\end{pmatrix},
\]
then 
\[
U_{T_{k+1}} = \begin{pmatrix}
    V_1\bvec{k}\\
    V_2\bvec{k}\\
    V_3\bvec{k}\\
\end{pmatrix} + \begin{pmatrix}
    A_{\rps} & B & -B\\
    -B & A_{\rps} & B\\
    B & -B & A_{\rps}
\end{pmatrix} \begin{pmatrix}
Q\bvec{k}    \\
0\\
0\\
\end{pmatrix} =  \begin{pmatrix}
    V_3\bvec{k+1}\\
    V_1\bvec{k+1}\\
    V_2\bvec{k+1}\\
\end{pmatrix}.
\]
Because of the cyclic symmetry of $M$, the same proof of Lemma~\ref{lem:oneblock} implies that
\[
U_{T_{k+2}} =  \begin{pmatrix}
    V_2\bvec{k+2}\\
    V_3\bvec{k+2}\\
    V_1\bvec{k+2}\\
\end{pmatrix}, \quad U_{T_{k+3}} = \begin{pmatrix}
    V_1\bvec{k+3}\\
    V_2\bvec{k+3}\\
    V_3\bvec{k+3}\\
\end{pmatrix}.
\]
For the case of $k=0$, notice that
\[
U_0 = U_{T_1} = \begin{pmatrix}
    V_1\bvec{1}\\
    V_2\bvec{1}\\
    V_3\bvec{1}\\
\end{pmatrix}.
\]
Induction on $k$ immediately implies that when $t=3k+1$,
\[
U_{t} = \begin{pmatrix}
    V_1\bvec{3k+1}\\
    V_2\bvec{3k+1}\\
    V_3\bvec{3k+1}\\
\end{pmatrix}.
\]
Finally, notice that when $t=T_{3k+1} = \Theta(k^3)$,
\[
 \max\{U_t\} = \begin{pmatrix}
    2 & \frac{406}{27} & 0
\end{pmatrix}^\top\bvec{3k+1} \ge 18k^2 = \Omega(t^{\frac{2}{3}}).
\]
Therefore
\[
\gap(x_t/t) = 2\max\{U_t\}/t = \Theta(t^{-\frac{1}{3}}).
\]

\end{proof}

\subsection{Numerical experiment}
The fact that a gap exists between the first and second largest entry except for the second step also means that the simulation of (\ref{eq:fp-sym}) can be easily carried out with floating point calculation. As shown in Fig.~\ref{fig:numerical}, the duality gap of $x_t/t$ indeed decreases at a rate of $\approx 0.36 t^{-\frac{1}{3}}$.

\begin{figure}[t]
    \centering
    \includegraphics[width=0.5\linewidth]{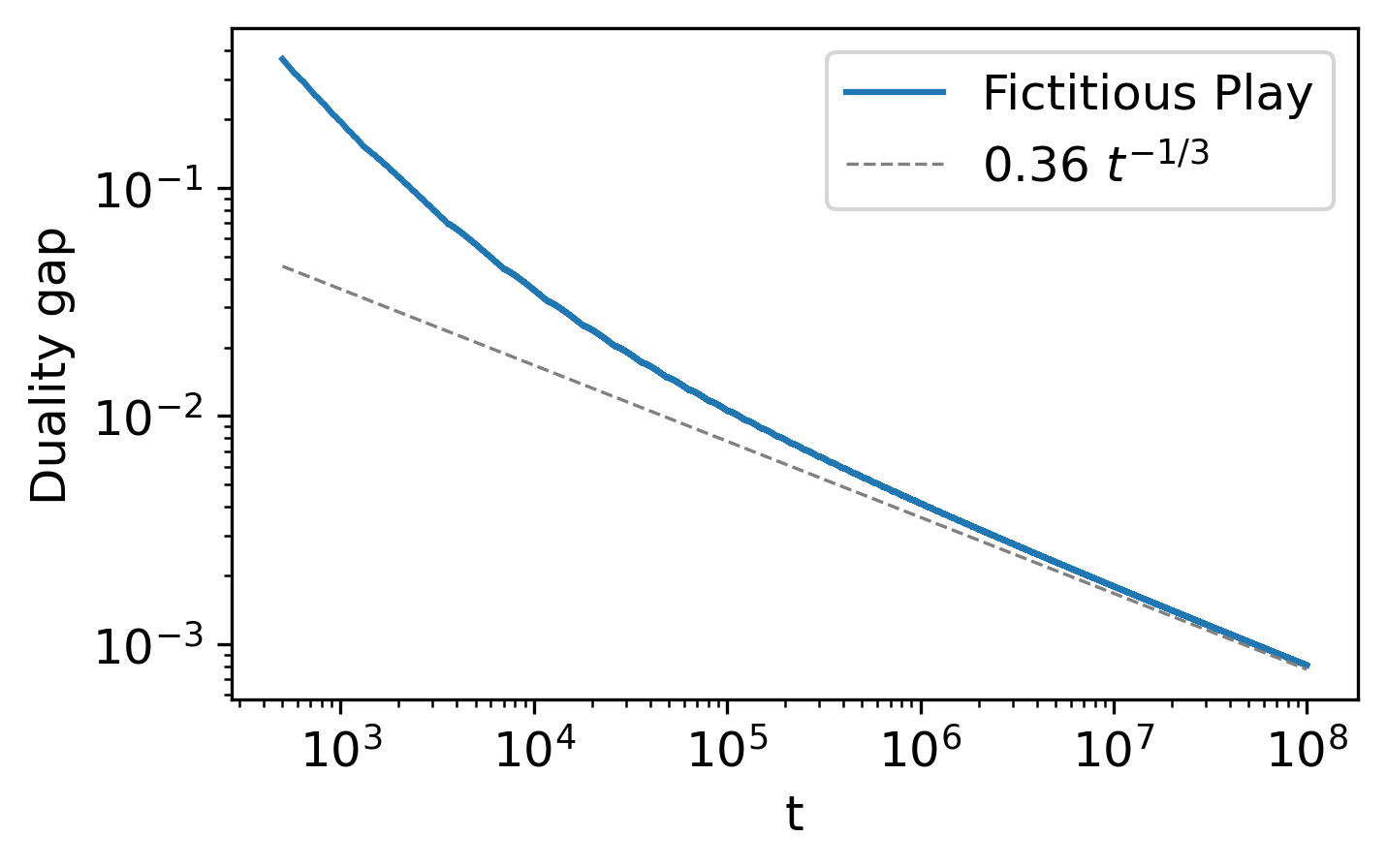}
    \caption{Duality gap of $x_t/t$ vs $t$ for FP in $M_{aug}$}
    \label{fig:numerical}
\end{figure}

\subsection*{Acknowledgements}
The author thanks Qiwen Cui for introducing him to the problem. The author thanks Chi Jin, Qinghua Liu and Zhou Lu for helpful discussions.

\bibliographystyle{plainnat}
\bibliography{ref}

\newpage
\appendix
\section{A Full Presentation of $M_{aug}$}
Choosing $\delta=\frac{1}{2700}$, the hard example is given by
\begin{equation*}
M_{aug}=\begin{pmatrix}
0 & -\frac{215689}{2700} & -\frac{15274}{225} & -\frac{215687}{2700} & -\frac{1}{1350} & -\frac{4087}{108} & -\frac{11329}{900} & -\frac{52063}{900} & -\frac{53897}{675} & -\frac{137287}{2700} \\
\frac{215689}{2700} & 0 & -1 & 1 & -\frac{71}{900} & -\frac{3}{50} & -\frac{1}{12} & \frac{71}{900} & \frac{3}{50} & \frac{1}{12} \\
\frac{15274}{225} & 1 & 0 & -1 & -\frac{3}{50} & -\frac{7}{300} & -\frac{1}{36} & \frac{3}{50} & \frac{7}{300} & \frac{1}{36} \\
\frac{215687}{2700} & -1 & 1 & 0 & -\frac{1}{12} & -\frac{1}{36} & -\frac{1}{18} & \frac{1}{12} & \frac{1}{36} & \frac{1}{18} \\
\frac{1}{1350} & \frac{71}{900} & \frac{3}{50} & \frac{1}{12} & 0 & -1 & 1 & -\frac{71}{900} & -\frac{3}{50} & -\frac{1}{12} \\
\frac{4087}{108} & \frac{3}{50} & \frac{7}{300} & \frac{1}{36} & 1 & 0 & -1 & -\frac{3}{50} & -\frac{7}{300} & -\frac{1}{36} \\
\frac{11329}{900} & \frac{1}{12} & \frac{1}{36} & \frac{1}{18} & -1 & 1 & 0 & -\frac{1}{12} & -\frac{1}{36} & -\frac{1}{18} \\
\frac{52063}{900} & -\frac{71}{900} & -\frac{3}{50} & -\frac{1}{12} & \frac{71}{900} & \frac{3}{50} & \frac{1}{12} & 0 & -1 & 1 \\
\frac{53897}{675} & -\frac{3}{50} & -\frac{7}{300} & -\frac{1}{36} & \frac{3}{50} & \frac{7}{300} & \frac{1}{36} & 1 & 0 & -1 \\
\frac{137287}{2700} & -\frac{1}{12} & -\frac{1}{36} & -\frac{1}{18} & \frac{1}{12} & \frac{1}{36} & \frac{1}{18} & -1 & 1 & 0
\end{pmatrix}.
\end{equation*}

\section{Omitted Proofs}
\subsection{Fact~\ref{fact:rps}}
   We can prove this inductively, where the hypothesis is that when $t=9k^2$, 
\[
x_t = [3k^2-2k, 3k^2, 3k^2 + 2k]^\top, U_t = [2k, -4k, 2k]^\top.
\]
By alphabetical tie-breaking, $i_{t+1}$ = $1$. Then
\[
U_{9k^2+1} = U_{9k^2} + A_{\rps}[:, 1] = [2k, -4k+1, 2k-1]^\top.
\]
As long as $i_t=1$, $U_{t}[1] = U_{t-1}[1]$, while $U_t[3] = U_{t-1}[3]-1$. $i_t$ only stops being the argmax when $U_{t-1}[2]>U_[t-1][1]$. Therefore, $i_t=1$ for $9k^2+1 \le t\le 9k^2 + 6k + 1$. When $t=9k^2+6k+1$, one has
\[
x_t = [3k^2+4k+1, 3k^2, 3k^2+2k]^\top, U_t = [2k, 2k+1, -4k-1].
\]
Similarly, $i_t=2$ for $9k^2+6k+2 \le t\le 9k^2+12k+4$. During this period, $U[2]$ stays constant, $U[1]$ decreases while $U[3]$ increases.  When $t=9k^2+12k+6$, one has
\[
x_t = [3k^2+4k+1, 3k^2+6k+3, 3k^2+2k]^\top, U_t = [-4k-3, 2k+1, 2k+2].
\]
Finally, when $9k^2+12k+5\le t \le 9k^2+18k+9$, $U[3]$ stays constant, $U[1]$ increases while $U[2]$ decreases. Thus, in this interval $i_t=3$. When $t=9k^2+18k+9=9(k+1)^2$,
\[
x_t = [3k^2+4k+1, 3k^2+6k+3, 3k^2+8k+5]^\top, U_t = [2k+2, -4k-4, 2k+2].
\]
The periodicity is then established via induction.

\subsection{Lemma~\ref{lem:init}}
It is assumed without loss of generality that $\hat{i}_1=1$. Therefore, $\hat{U}_1=[0;\hat{U}_0]$. Since
\[
\max\{\hat{U}_0\} > \min\{\hat{U}_0\} > 0, 
\]
1. when an action in $[2,10]$ is played, $\hat{U}[1]$ decreases; 2. $\max \{\hat{U}_1\}[2:10]>0=\hat{U_1}[1]$. Therefore, starting from $t=1$, 
\[
{\hat{U}_t[1]} < 0 < \max_t\{\hat{U}\},
\]
which means that action $1$ will never be played again. Therefore, the action sequence $\hat{i_{t+1}}$ will be identical to $i_t$ (action sequence of game $M$ initialized with $\hat{U}_0$).

Next, $\hat{U}_0$ is different from $U_0$ in only two ways: 1. an additive constant which makes no difference to FP; 2. a small additive constant $\delta$. In $M$, ties only occur within block, and the lower bound assumption is that the ties are broken lexicographically. By adding a small constant $\delta$, no ties will occur in $M$, and the within-block trajectory is identical to RPS with lexicographic tie-breaking. Since $\delta < 1/1800$ while the largest denominator in $M$ is $900$, $2\delta$ will be smaller than existing utility gaps.

\end{document}